\documentclass[english]{scrartcl}
\usepackage[T1]{fontenc}
\usepackage[latin9]{inputenc}
\usepackage{mathtools}
\usepackage{amsmath}
\usepackage{amssymb}

\makeatletter
 \usepackage{amsthm}
 \theoremstyle{plain}
\newtheorem{thm}{Theorem}[section]
\newcommand{\lyxaddress}[1]{
\par {\raggedright #1
\vspace{1.4em}
\noindent\par}
}
  \theoremstyle{plain}
  \newtheorem{lem}[thm]{Lemma}
  \theoremstyle{remark}
  \newtheorem{rem}[thm]{Remark}

\usepackage{mathtools}

\newcommand*{\grad}{\operatorname{grad}}

\newcommand*{\dive}{\operatorname{div}}
\newcommand*{\Grad}{\operatorname{Grad}}

\newcommand*{\Div}{\operatorname{Div}}

\newcommand*{\ii}{\mathrm{i}}

\DeclareMathAccent{\Circ}{\mathalpha}{operators}{"17}
\newcommand{\interior}[1]{\Circ{#1}}

\renewcommand{\Re}{\operatorname{\mathfrak{Re}}}

\renewcommand{\tilde}{\widetilde}
\renewcommand*{\epsilon}{\varepsilon}
\renewcommand*{\rho}{\varrho}

\makeatother

\usepackage{babel}
\begin{document}

\title{On an Elasto-Acoustic Transmission Problem in Anisotropic, Inhomogeneous
Media.}

\author{R.~Picard}
\maketitle

\lyxaddress{\begin{center}
Department of Mathematics, TU Dresden, 01062 Dresden, Germany
\par\end{center}}
\begin{abstract}
\textbf{Abstract.} We consider a coupled system describing the interaction
between acoustic and elastic regions, where the coupling occurs not
via material properties but through an  interaction on an interface
separating the two regimes. Evolutionary well-posedness in the sense
of Hadamard well-posedness supplemented by causal dependence is shown
for a natural choice of generalized interface conditions. The results
are obtained in a real Hilbert space setting incurring no regularity
constraints on the boundary and almost none on the interface of the
underlying regions.
\end{abstract}

\section{Introduction}

Similarities between various initial boundary value problems of mathematical
physics have been noted as general observations throughout the literature.
Indeed, the work by K. O. Friedrichs, \cite{Ref166,CPA:CPA3160110306},
already showed that the classical linear phenomena of mathematical
physics belong \textendash{} in the static case \textendash{} to his
class of \emph{symmetric positive hyperbolic partial differential
equations}, later referred to as \emph{Friedrichs systems}, which
are of the abstract form

\begin{equation}
\left(M_{1}+A\right)u=f,\label{eq:pssym}
\end{equation}

with $A$ at least formally, i.e. on $C_{\infty}$-vector fields with
compact support in the underlying region $\Omega$, a skew-symmetric
differential operator and the $L^{\infty}$-matrix-valued multiplication-operator
$M_{1}$ satisfying the condition 
\[
\mathrm{sym}\left(M_{1}\right)\coloneqq\frac{1}{2}\left(M_{1}+M_{1}^{*}\right)\geq c>0
\]
for some real number $c$. Indeed, a typical choice of boundary condition
is, when $A$ is skew-selfadjoint\footnote{To assume $A$ to be skew-selfadjoint is less restrictive than one
might think. For this we note that for example typical dissipative
boundary conditions actually give rise to natural skew-selfadjoint
spatial operators $A$, \cite{Picard20164888}. That skew-selfadjoint
$A$ is a quite common assumption but may not be recognized. As a
typical example we consider the popular transcription of the wave
equation $\partial_{0}^{2}-\Delta_{D}$, where $\Delta_{D}$ denotes
the Laplacian with a homogeneous Dirichlet boundary condition in a
bounded domain $\Omega$, into a first order system of the form $\partial_{0}+A$,
where $A=\left(\begin{array}{cc}
0 & \Delta_{D}\\
1 & 0
\end{array}\right)$ is indeed skew-selfadjoint due to the standard choice of Hilbert
space setting. } ($A$ m-accretive would be sufficient). Problem (\ref{eq:pssym})
can be considered as the static problem associated with the dynamic
problem ($\partial_{0}$ denotes the time-derivative) 
\begin{equation}
\partial_{0}M_{0}+M_{1}+A\label{eq:FS}
\end{equation}

with $M_{0}$ selfadjoint $L^{\infty}$-multiplication-operator and
$M_{0}\geq0$, which were also addressed in \cite{CPA:CPA3160110306}.
It is noteworthy, that even the temporal exponential weight factor,
which plays a central role in the approach of \cite{Pi2009-1}, is
introduced as an ad-hoc formal trick to produce a suitable $M_{1}$
for a well-posed static problem. For the so-called time-harmonic case,
where $\partial_{0}$ is replaced by $\ii\omega$, $\omega\in\mathbb{R}$,
we replace $A$ simply by $\ii\omega M_{0}+A$ to arrive at a system
of the form (\ref{eq:pssym}).

Operators of the Friedrichs type (\ref{eq:FS}), can be generalized
to obtain a fully time-dependent theory allowing for operator-valued
coefficients, indeed, in the time-shift invariant case, for systems
of the general form 
\begin{equation}
\tag{{Evo-Sys}}\left(\partial_{0}M\left(\partial_{0}^{-1}\right)+A\right)U=F\label{eq:problem-2}
\end{equation}
where $A$ is \textendash{} for simplicity \textendash{} skew-selfadjoint
and $M$ an operator-valued \textendash{} say \textendash{} rational
function (regular at 0) as an abstract coefficient. The meaning of
the so-called material law operator $M\left(\partial_{0}^{-1}\right)$
is in terms of a suitable function calculus associated with the (normal)
operator $\partial_{0}$, \cite[Chapter 6]{PDE_DeGruyter}. This spacious
class of operators allows for a large class of material laws including
\textendash{} the recently of great interest \textendash{} \emph{meta-materials}. 

We shall refer to such systems as evo-systems (or evolutionary equations)
to distinguish them from the special subclass of classical (explicit)
evolution equations.

In this paper we intend to study a particular transmission problem
between two physical regimes, acoustics and elasto-dynamics, within
this general framework to establish its well-posedness, which for
evo-systems entails not only Hadamard well-posedness, i.e. \emph{uniqueness},
\emph{existence} and \emph{continuous dependence}, but also the crucial
property of \emph{causality}. For this we will only have to establish
the skew-selfadjointness of a suitably constructed operator $A$.
Then it is known that the requirement
\begin{equation}
\rho M\left(0\right)+\mathrm{sym}\left(M^{\prime}\left(0\right)\right)\geq c_{0}>0\label{eq:posdef-rat}
\end{equation}
for some number $c_{0}$ all sufficiently large $\rho\in]0,\infty[\,,$
yields the desired well-posedness, see the survey \cite{Picard2015}.
For the simple Friedrichs type case where we additionally assume 
\begin{equation}
M_{0}=M\left(0\right)\geq c_{0}>0\label{eq:posdef-M0}
\end{equation}
for some number $c_{0}\,,$ which clearly implies (\ref{eq:posdef-rat}),
we may even use the commonly invoked semi-group theory to establish
the desired well-posedness (note that in this case $M_{1}=M^{\prime}\left(0\right)$
and all higher derivatives of $M$ vanish). Indeed, under these strong
restrictions (\ref{eq:FS}) is congruent to
\begin{equation}
\partial_{0}+\sqrt{M_{0}^{-1}}M_{1}\sqrt{M_{0}^{-1}}+\sqrt{M_{0}^{-1}}A\sqrt{M_{0}^{-1}},\label{eq:sg-form}
\end{equation}
which amounts to having $M_{0}=1$ ($M_{1}$ replaced by the congruent
$\sqrt{M_{0}^{-1}}M_{1}\sqrt{M_{0}^{-1}}$) and using $\sqrt{M_{0}}U$
as the new unknown in the corresponding problem of the form (\ref{eq:problem-2}).
With $\sqrt{M_{0}^{-1}}A\sqrt{M_{0}^{-1}}$ inheriting its skew-selfadjointness
from $A$ we obtain indeed a one-parameter group $\left(\exp\left(t\sqrt{M_{0}^{-1}}A\sqrt{M_{0}^{-1}}\right)\right)_{t\in\mathbb{R}}$,
which by a simple perturbation argument yields a group $\left(U\left(t\right)\right)_{t\in\mathbb{R}}$
such that $\left(\chi_{_{[0,\infty[}}\left(t\right)U\left(t\right)\right)_{t\in\mathbb{R}}$,
with $\chi_{_{[0,\infty[}}$ denoting the characteristic function
of the interval $[0,\infty[\,,$ is the fundamental solution associated
with (\ref{eq:sg-form}). Thus a fairly general solution can be obtained
by convolution with this fundamental solution. Restricting this fundamental
solution to its support yields a continuous, one-parameter semigroup
$\left(U\left(t\right)\right)_{t\in[0,\infty[}\,.$ In any case we
are justified to focus on the underlying skew-selfadjointness of the
operator $A$ as a central feature to obtain well-posedness for a
large class of general material laws, since we shall be concerned
with the interaction between the elastic and the acoustic regimes
solely via the interface, not via material interactions through the
material law, as for example in piezo-electrics, compare e.g. \cite{MMA:MMA3866}
for a typical effect of the latter type. This specific focus also
allows us in the interest of brevity to by-pass the intricacies of
the time-dependent theory of \cite{Pi2009-1}.

Skew-selfadjointness of an operator $A$, i.e.
\begin{equation}
A=-A^{*},\label{eq:skew}
\end{equation}
 in a real Hilbert space $H$ results in
\[
\left\langle u|Au\right\rangle _{H}=0
\]
for all $u\in D\left(A\right)$. Moreover, in typical cases skew-selfadjointness
of $A$ is a simple consequence of $A$ being congruent to a block
matrix of the form 
\[
\left(\begin{array}{cc}
0 & -C^{*}\\
C & 0
\end{array}\right),
\]
where $C:D\left(C\right)\subseteq H_{0}\to H_{1}$ is a closed, densely
defined, linear operator between real Hilbert spaces $H_{0},\:H_{1}$,
which is clearly skew-selfadjoint in the direct sum Hilbert space
$H=H_{0}\oplus H_{1}$. 

The interest of studying the coupling between acoustic and elasticity
wave phenomena has a relatively long history in the engineering community,
with \cite{10.2307/93996}, \cite{doi:10.1121/1.1916256}, being earlier
references. Originally motivated by submarine noise propagation, this
coupling is also of interest in connection with loudspeaker and hearing
aid design, as well as non-destructive testing. Near the close of
the last century there has been a rekindled interest in these specific
issues, \cite{doi:10.1121/1.397156}, \cite{Luke:1995:FIA:214875.214884}.
More recent publications are the numerical investigations \cite{NME:NME1669},
\cite{WILCOX20109373}, \cite{SannaMonkola2011}, and the more mathematically
oriented \cite{0963.35043}, \cite{zbMATH05185539}, \cite{KANG2017686},
\cite{doi:10.1137/16M1090326}, just to mention a few. Here we want
to transcend the predominant constant coefficient and \textendash{}
with the notable exception of \cite{doi:10.1137/16M1090326} \textendash{}
largely time-harmonic analysis and consider the time-dependent case
in anisotropic, inhomogeneous media. Since we shall consider operator
coefficients, this also includes media with non-local behavior. For
sake of accessibility we restrict our attention to the autonomous
case with classical block-diagonal material laws and no memory effects.
We use a functional-analytical setting in real Hilbert space to obtain
a well-posedness for this elasto-acoustic transmission problem.

We shall first establish the spatial operator of acoustics and elasticity,
respectively, as intimately related skew-selfadjoint operators (mother-descendant
mechanism) in a real Hilbert space framework based on the above-mentioned
block structure with suitably introduced operators $C$. Then, in
Section \ref{sec:An-Interface-Coupling} we apply these observations
to a particular interface coupling problem between the two regimes
in adjacent regions via a refined mother-descendant mechanism. We
emphasize that our setup allows for arbitrary open sets as underlying
domains with no additional constraints on boundary regularity and
almost no constraints on interface regularity. Indeed we only require
the interface to be a Lebesgue null set. The induced homogeneous boundary
value constraints and transmission conditions are encoded \textendash{}
as customary \textendash{} in suitable generalization as containment
in the domain of the operator.

\section{The Connection of the Spatial Operators of Acoustics and Elasticity}

\subsection{Basic Ideas}

Without loss of generality we may and will assume that all Hilbert
spaces used in the following are real\footnote{Note that every complex Hilbert space $X$ is a real Hilbert space
choosing only real numbers as multipliers and 
\[
\left(\phi,\psi\right)\mapsto\Re\left\langle \phi|\psi\right\rangle _{X}
\]
as new inner product. Note that with this choice $\phi$ and $\ii\phi$
are always orthogonal. Moreover, for any skew-symmetric operator $A$
we have
\[
x\perp Ax
\]
for all $x\in D\left(A\right)$. 

Indeed, since $\left\langle x|y\right\rangle -\left\langle y|x\right\rangle =0$
(symmetry) we have
\[
\left\langle x|Ax\right\rangle -\left\langle Ax|x\right\rangle =0
\]
or by skew-symmetry
\begin{eqnarray*}
0 & = & \left\langle x|Ax\right\rangle -\left\langle Ax|x\right\rangle \\
 & = & 2\left\langle x|Ax\right\rangle 
\end{eqnarray*}
for all $x\in D\left(A\right)$. }.

In many practical cases the desired skew-selfadjointness of the spatial
operator $A$ is evident from its structure as a block operator matrix
of the form
\[
A=\left(\begin{array}{cc}
0 & -C^{*}\\
C & 0
\end{array}\right),
\]
with $H=H_{0}\oplus H_{1}$ and $C:D\left(C\right)\subseteq H_{0}\to H_{1}$
a closed, densely defined, linear operator. We shall start our exploration
by focussing for simplicity and definiteness on the Cartesian situation
and on the case of the so-called Dirichlet boundary condition. For
this, we initially take $C$ as the closure $\interior\grad$ of the
classical differential operator
\begin{align*}
\interior C_{1}\left(\Omega,\mathbb{R}^{3}\right)\subseteq L^{2}\left(\Omega,\mathbb{R}^{3}\right) & \to L^{2}\left(\Omega,\mathbb{R}^{3\times3}\right),\\
u & \mapsto u^{\prime},
\end{align*}
where $u^{\prime}$ is the derivative (in matrix language the Jacobian)
of the vectorfield $u$. The negative adjoint is the weak extension
of the classical divergence operator on matrix fields
\[
\dive\coloneqq-\left(\interior\grad\right)^{*}.
\]
Thus, the operator of our initial interest is 
\[
A=\left(\begin{array}{cc}
0 & \dive\\
\interior\grad & 0
\end{array}\right)
\]
as a skew-selfadjoint operator in $H=L^{2}\left(\Omega,\mathbb{R}^{3}\right)\oplus L^{2}\left(\Omega,\mathbb{R}^{3\times3}\right).$
Here $\mathbb{R}^{3\times3}$ is equipped with the standard Frobenius
inner product. As an illustration let us consider
\[
\left(\partial_{0}\left(\begin{array}{cc}
\rho_{*} & 0\\
0 & C^{-1}
\end{array}\right)+\left(\begin{array}{cc}
0 & -\dive\\
-\interior\grad & 0
\end{array}\right)\right)\left(\begin{array}{c}
v\\
T
\end{array}\right)=\left(\begin{array}{c}
f\\
g
\end{array}\right)
\]
as an associated dynamic problem for finding a solution $\left(\begin{array}{c}
v\\
T
\end{array}\right)\in L^{2}\left(\Omega,\mathbb{R}^{3}\right)\oplus L^{2}\left(\Omega,\mathbb{R}^{3\times3}\right).$ Here $\rho_{*}:L^{2}\left(\Omega,\mathbb{R}^{3}\right)\to L^{2}\left(\Omega,\mathbb{R}^{3}\right)$,
and $C:L^{2}\left(\Omega,\mathbb{R}^{3\times3}\right)\to L^{2}\left(\Omega,\mathbb{R}^{3\times3}\right)$
are assumed to be strongly positive definite mappings in order to
obtain well-posedness in the sense of our introductory exposition.
This type of system can be understood as modeling asymmetric elasticity
theory in the sense of \cite{zbMATH03315043,zbMATH03420318,Nowacki1986}. 

\subsection{Symmetric Elasticity as a Descendant of Asymmetric Elasticity.}

To illustrate the mother-descendant mechanism, as introduced in \cite{Mother2012,Mother2013},
see also \cite{ZAMM:ZAMM201300297}, we first perform the transition
to classical (symmetric) elasticity using this concept. 

We recall from \cite{zbMATH06250993} the following simple but crucial
lemma.
\begin{lem}
Let $C:D\left(C\right)\subseteq H\to Y$ be a closed densely-defined
linear operator between Hilbert spaces $H,\:Y$. Moreover, let $B:Y\to X$
be a continuous linear operator into another Hilbert space $X$. If
$C^{*}B^{*}$ is densely defined, then
\[
\overline{BC}=\left(C^{*}B^{*}\right)^{*}.
\]
\end{lem}
\begin{proof}
It is 
\[
C^{*}B^{*}\subseteq\left(BC\right)^{*}.
\]
If $\phi\in D\left(\left(BC\right)^{*}\right)$ then 
\[
\left\langle BCu|\phi\right\rangle _{X}=\left\langle u|\left(BC\right)^{*}\phi\right\rangle _{H}
\]
for all $u\in D\left(C\right)$. Thus, we have
\[
\left\langle Cu|B^{*}\phi\right\rangle _{Y}=\left\langle BCu|\phi\right\rangle _{X}=\left\langle u|\left(BC\right)^{*}\phi\right\rangle _{H}
\]
for all $u\in D\left(C\right)$ and we read off that $B^{*}\phi\in D\left(C^{*}\right)$
and 
\[
C^{*}B^{*}\phi=\left(BC\right)^{*}\phi.
\]
Thus we have
\[
\left(BC\right)^{*}=C^{*}B^{*}.
\]
If now $C^{*}B^{*}$ is densely defined, we have for its adjoint operator
\[
\left(C^{*}B^{*}\right)^{*}=\overline{BC}.
\]
\end{proof}
As a consequence we have that the \emph{descendant}
\[
\overline{\left(\begin{array}{cc}
1 & 0\\
0 & B
\end{array}\right)\left(\begin{array}{cc}
0 & -C^{*}\\
C & 0
\end{array}\right)}\left(\begin{array}{cc}
1 & 0\\
0 & B^{*}
\end{array}\right)=\left(\begin{array}{cc}
0 & -C^{*}B^{*}\\
\overline{BC} & 0
\end{array}\right)
\]
indeed inherits its skew-selfadjointness from its \emph{mother} $\left(\begin{array}{cc}
0 & -C^{*}\\
C & 0
\end{array}\right)$ (with $C$ replaced by $\overline{BC}$).
\begin{rem}
Clearly, the role of the components can be interchanged so that
\[
\overline{\left(\begin{array}{cc}
D & 0\\
0 & 1
\end{array}\right)\left(\begin{array}{cc}
0 & -C^{*}\\
C & 0
\end{array}\right)}\left(\begin{array}{cc}
D^{*} & 0\\
0 & 1
\end{array}\right)=\left(\begin{array}{cc}
0 & -DC^{*}\\
CD^{*} & 0
\end{array}\right)
\]
with $D:H\to Y$ such that $CD^{*}$ is densely defined, is also a
valid descendant construction.

These construction can be combined. In general, a repeated application
of the mother-descendant mechanism may, however, depend on the order
in which they are carried out. This fact has been overlooked in \cite{Mother2013}.
An illuminating example is choosing $C$ as the weak $L^{2}\left(\mathbb{R}\right)-$derivative
$\partial$ and $B=D$ as the cut-off by the characteristic function
$\chi_{_{]-1/2,1/2[}}$ of the symmetric unit interval $]-1/2,1/2[$
yielding
\begin{equation}
\left(\begin{array}{cc}
0 & \overline{\chi_{_{]-1/2,1/2[}}\left(\partial\chi_{_{]-1/2,1/2[}}\right)}\\
\overline{\chi_{_{]-1/2,1/2[}}\partial}\chi_{_{]-1/2,1/2[}} & 0
\end{array}\right)\label{eq:D2}
\end{equation}
if first the construction with $B$ and then with $D$ is carried
out. In reverse order we obtain
\begin{equation}
\left(\begin{array}{cc}
0 & \overline{\chi_{_{]-1/2,1/2[}}\partial}\chi_{_{]-1/2,1/2[}}\\
\overline{\chi_{_{]-1/2,1/2[}}\left(\partial\chi_{_{]-1/2,1/2[}}\right)} & 0
\end{array}\right).\label{eq:D1}
\end{equation}
In comparison (\ref{eq:D2}) models vanishing at $\pm\frac{1}{2}$
for the second component, whereas (\ref{eq:D1}) leads to vanishing
at $\pm\frac{1}{2}$ of the first component. 
\end{rem}
As a convenient mother operator to start from we take the above-mentioned
theory of asymmetric elasticity of Nowacki, \cite{zbMATH03315043,Nowacki1986}.
Indeed, classical (symmetric) elasticity theory can be considered
as a descendant in the above sense of the form
\begin{align}
 & \left(\begin{array}{cc}
0 & -\Div\\
-\interior\Grad & 0
\end{array}\right),\label{eq:sym-elast}
\end{align}
where 
\[
\interior\Grad\coloneqq\overline{\iota_{\mathrm{sym}}^{*}\interior\grad}
\]
and 
\[
\Div\coloneqq\dive\iota_{\mathrm{sym}}
\]
with 
\begin{align*}
\iota_{\mathrm{sym}}:L^{2}\left(\Omega,\mathrm{sym}\left[\mathbb{R}^{3\times3}\right]\right) & \to L^{2}\left(\Omega,\mathbb{R}^{3\times3}\right),\\
T & \mapsto T,
\end{align*}
where $\mathrm{sym}\left[\mathbb{R}^{3\times3}\right]$ denotes the
image of $\mathbb{R}^{3\times3}$ under the mapping $\mathrm{sym}$,
i.e. we have in the descendant construction $B=\iota_{\mathrm{sym}}^{*}$.
Note that 
\[
\iota_{\mathrm{sym}}^{*}T=\mathrm{sym}\left(T\right)
\]
for all $T\in L^{2}\left(\Omega,\mathbb{R}^{3\times3}\right)$.

\subsection{Acoustics as a Descendant of Asymmetric Elasticity.}

The spatial operator used in the acoustics model can also be introduced
as a descendant of asymmetric elasticity. It is actually the scalar
version corresponding to the asymmetric elasticity case.

Indeed, classical acoustics can be considered as a descendant of the
form
\begin{align*}
 & \left(\begin{array}{cc}
0 & \grad\\
\interior\dive & 0
\end{array}\right),
\end{align*}
where we re-use the classical notations by letting
\[
\interior\dive\coloneqq\overline{\mathrm{trace}\;\interior\grad}
\]
and 
\[
\grad\coloneqq\dive\:\mathrm{trace}^{*}
\]
with 
\begin{align*}
\mathrm{trace}:L^{2}\left(\Omega,\mathbb{R}^{3\times3}\right) & \to L^{2}\left(\Omega,\mathbb{R}\right),\\
T=\left(T_{ij}\right)_{i,j} & \mapsto\mathrm{trace}\:T\coloneqq\sum_{i}T_{ii},
\end{align*}
i.e. $B=\mathrm{trace}$. Note that 
\[
\mathrm{trace}^{*}p=\left(\begin{array}{ccc}
p & 0 & 0\\
0 & p & 0\\
0 & 0 & p
\end{array}\right)
\]
for all $p\in L^{2}\left(\Omega,\mathbb{R}\right)$.
\begin{rem}
The acoustic system can also be constructed by applying $B=\mathrm{trace}$
to the symmetric elasticity operator (\ref{eq:sym-elast}). Note that
the pressure distribution $p$ is in both cases obtained from the
stress as 
\[
p\coloneqq-\mathrm{trace}\,T.
\]
\end{rem}

\section{\label{sec:An-Interface-Coupling}An Interface Coupling Between Acoustics
and Elasticity}

We will now combine the two descendant constructions above to obtain
an interface coupling set-up for the skew-selfadjoint operator $A$.
We assume $\Omega_{0}\cup\Omega_{1}\subseteq\Omega$, such that the
orthogonal decomposition\footnote{Consequently, we also have 
\begin{eqnarray*}
L^{2}\left(\Omega,\mathbb{R}^{3\times3}\right) & = & L^{2}\left(\Omega_{0},\mathbb{R}^{3\times3}\right)\oplus L^{2}\left(\Omega_{1},\mathbb{R}^{3\times3}\right),\\
L^{2}\left(\Omega,\mathbb{R}^{3}\right) & = & L^{2}\left(\Omega_{0},\mathbb{R}^{3}\right)\oplus L^{2}\left(\Omega_{1},\mathbb{R}^{3}\right).
\end{eqnarray*}
} 
\begin{eqnarray}
L^{2}\left(\Omega,\mathbb{R}\right) & = & L^{2}\left(\Omega_{0},\mathbb{R}\right)\oplus L^{2}\left(\Omega_{1},\mathbb{R}\right)\label{eq:deco-v}
\end{eqnarray}
holds.

Then, with the respective canonical embeddings into $L^{2}\left(\Omega,\mathbb{R}^{3\times3}\right)$
we obtain 
\begin{eqnarray*}
B:L^{2}\left(\Omega,\mathbb{R}^{3\times3}\right) & \to & L^{2}\left(\Omega_{0},\mathrm{sym}\left[\mathbb{R}^{3\times3}\right]\right)\oplus L^{2}\left(\Omega_{1},\mathbb{R}\right),\\
T & \mapsto & \left(\begin{array}{c}
\iota_{L^{2}\left(\Omega_{0},\mathrm{sym}\left[\mathbb{R}^{3\times3}\right]\right)}^{*}\iota_{\mathrm{sym}}^{*}T\\
-\iota_{L^{2}\left(\Omega_{1},\mathbb{R}\right)}^{*}\mathrm{trace}\:T
\end{array}\right),
\end{eqnarray*}
and so

\[
B=\left(\begin{array}{c}
\iota_{L^{2}\left(\Omega_{0},\mathrm{sym}\left[\mathbb{R}^{3\times3}\right]\right)}^{*}\iota_{\mathrm{sym}}^{*}\\
-\iota_{L^{2}\left(\Omega_{1},\mathbb{R}\right)}^{*}\;\mathrm{trace}
\end{array}\right).
\]

With  this we get as a descendant construction
\begin{eqnarray}
A & = & \overline{\left(\begin{array}{cc}
1 & 0\\
0 & B
\end{array}\right)\left(\begin{array}{cc}
0 & -\dive\\
-\interior\grad & 0
\end{array}\right)}\left(\begin{array}{cc}
1 & 0\\
0 & B^{*}
\end{array}\right)\label{eq:incl}\\
 &  & \subseteq\left(\begin{array}{cc}
0 & \left(\begin{array}{cc}
-\Div_{\Omega_{0}} & -\grad_{\Omega_{1}}\end{array}\right)\\
\left(\begin{array}{c}
-\Grad_{\Omega_{0}}\\
\dive_{\Omega_{1}}
\end{array}\right) & \left(\begin{array}{cc}
0 & 0\\
0 & 0
\end{array}\right)
\end{array}\right)\label{eq:incl-1}
\end{eqnarray}
and 
\begin{align*}
M\left(\partial_{0}^{-1}\right)=M\left(0\right) & =\left(\begin{array}{cc}
\rho_{*,\Omega_{0}}+\kappa_{\Omega_{1}}^{-1} & \left(\begin{array}{cc}
0 & 0\end{array}\right)\\
\left(\begin{array}{c}
0\\
0
\end{array}\right) & \left(\begin{array}{cc}
C_{\Omega_{0}}^{-1} & 0\\
0 & c_{\Omega_{1}}
\end{array}\right)
\end{array}\right).
\end{align*}
The indices $\Omega_{k}$, $k=0,1,$ are used to denote the respective
supports of the quantities. The unknowns are now of the form 
\[
\left(\begin{array}{c}
v_{\Omega_{0}}+v_{\Omega_{1}}\\
\left(\begin{array}{c}
T_{\Omega_{0}}\\
p_{\Omega_{1}}
\end{array}\right)
\end{array}\right)\in H=L^{2}\left(\Omega,\mathbb{R}^{3}\right)\oplus\left(L^{2}\left(\Omega_{0},\mathrm{sym}\left[\mathbb{R}^{3\times3}\right]\right)\oplus L^{2}\left(\Omega_{1},\mathbb{R}\right)\right),
\]
where the first component is to be understood in the sense of (\ref{eq:deco-v}).
From the inclusion (\ref{eq:incl}),(\ref{eq:incl-1}), we read off
that the resulting evo-system 
\begin{equation}
\left(\partial_{0}M\left(0\right)+A\right)\left(\begin{array}{c}
v_{\Omega_{0}}+v_{\Omega_{1}}\\
\left(\begin{array}{c}
T_{\Omega_{0}}\\
p_{\Omega_{1}}
\end{array}\right)
\end{array}\right)=\left(\begin{array}{c}
f_{\Omega_{0}}+f_{\Omega_{1}}\\
\left(\begin{array}{c}
F_{\Omega_{0}}\\
g_{\Omega_{1}}
\end{array}\right)
\end{array}\right)\label{eq:evo-mix}
\end{equation}
indeed yields
\[
\partial_{0}\left(\rho_{*,\Omega_{0}}+\kappa_{\Omega_{1}}^{-1}\right)\left(v_{\Omega_{0}}+v_{\Omega_{1}}\right)+\Div_{\Omega_{0}}T_{\Omega_{0}}+\grad_{\Omega_{1}}p_{\Omega_{1}}=f_{\Omega_{0}}+f_{\Omega_{1}},
\]
which in turn \textendash{} according to (\ref{eq:deco-v}) \textendash{}
splits into equations in $\Omega_{0}$ and in $\Omega_{1}$
\begin{eqnarray*}
\partial_{0}\rho_{*,\Omega_{0}}v_{\Omega_{0}}-\Div_{\Omega_{0}}T_{\Omega_{0}} & = & f_{\Omega_{0}},\\
\partial_{0}\kappa_{\Omega_{1}}^{-1}v_{\Omega_{1}}+\grad_{\Omega_{1}}p_{\Omega_{1}} & = & f_{\Omega_{1}}.
\end{eqnarray*}
The second block row yields another pair of equations
\begin{eqnarray*}
\partial_{0}C^{-1}T_{\Omega_{0}}-\Grad v_{\Omega_{0}} & = & F_{\Omega_{0}},\\
\partial_{0}c_{\Omega_{1}}p_{\Omega_{1}}+\dive v_{\Omega_{1}} & = & g_{\Omega_{1}}.
\end{eqnarray*}
The actual system models now generalize natural transmission conditions
on the common boundary part $\dot{\Omega}_{0}\cap\dot{\Omega}_{1}$
and the homogeneous Dirichlet boundary condition on $\dot{\Omega}_{0}\setminus\dot{\Omega}_{1}$
and the standard homogeneous Neumann boundary condition on $\dot{\Omega}_{1}\setminus\dot{\Omega}_{0}$
without assuming any smoothness of the boundary via containment of
the solution $U=\left(\begin{array}{c}
v_{\Omega_{0}}+v_{\Omega_{1}}\\
\left(\begin{array}{c}
T_{\Omega_{0}}\\
p_{\Omega_{1}}
\end{array}\right)
\end{array}\right)$ in the operator domain $D\left(\overline{\partial_{0}M\left(0\right)+A}\right)$.
Since we do not have maximal regularity in this case, this does not
mean that $U\in D\left(A\right)$, but we do have
\[
\partial_{0}^{-1}U\in D\left(A\right)
\]
as a form of expressing \emph{generalized} boundary constraints and
and transmission conditions. 

If, however, we assume sufficient regularity of the boundary and solution
one can easily motivate that the model yields a generalization of
classical transmission conditions on $\dot{\Omega}_{0}\cap\dot{\Omega}_{1}$.
Indeed, with 
\[
\left(\begin{array}{c}
v_{\Omega_{0}}+v_{\Omega_{1}}\\
\left(\begin{array}{c}
T_{\Omega_{0}}\\
p_{\Omega_{1}}
\end{array}\right)
\end{array}\right)\in D\left(A\right)
\]
we have (noting for the smooth exterior unit normal vector fields
$n_{\dot{\Omega}_{0}}$, $n_{\dot{\Omega}_{1}}$ on the boundaries
of $\Omega_{0}$ and $\Omega_{1}$, respectively, that $n_{\dot{\Omega}_{0}}=-n_{\dot{\Omega}_{1}}$
on $\dot{\Omega}_{0}\cap\dot{\Omega}_{1}$) with 
\[
\tilde{A}=\left(\begin{array}{cc}
0 & \left(-\begin{array}{cc}
\Div_{\Omega_{0}} & \grad_{\Omega_{1}}\end{array}\right)\\
\left(\begin{array}{c}
-\Grad_{\Omega_{0}}\\
\dive_{\Omega_{1}}
\end{array}\right) & \left(\begin{array}{cc}
0 & 0\\
0 & 0
\end{array}\right)
\end{array}\right),
\]
\begin{eqnarray*}
0 & = & \left\langle \left(\begin{array}{c}
v_{\Omega_{0}}+v_{\Omega_{1}}\\
\left(\begin{array}{c}
T_{\Omega_{0}}\\
p_{\Omega_{1}}
\end{array}\right)
\end{array}\right)\Big|\tilde{A}\left(\begin{array}{c}
v_{\Omega_{0}}+v_{\Omega_{1}}\\
\left(\begin{array}{c}
T_{\Omega_{0}}\\
p_{\Omega_{1}}
\end{array}\right)
\end{array}\right)\right\rangle \\
 & = & -\left\langle v_{\Omega_{0}}\Big||\:\Div\:T_{\Omega_{0}}\right\rangle -\left\langle T_{\Omega_{0}}\Big|\Grad_{\Omega_{0}}v_{\Omega_{0}}\right\rangle +\\
 &  & +\left\langle p_{\Omega_{1}}\Big|\dive_{\Omega_{1}}v_{\Omega_{1}}\right\rangle +\left\langle v_{\Omega_{1}}\Big|\grad_{\Omega_{1}}p_{\Omega_{1}}\right\rangle \\
 & = & -\int_{\dot{\Omega}_{0}\cap\dot{\Omega}_{1}}v_{\Omega_{0}}^{\top}T_{\Omega_{0}}n_{\dot{\Omega}_{0}}do+\int_{\dot{\Omega}_{0}\cap\dot{\Omega}_{1}}n_{\dot{\Omega}_{1}}^{\top}\left(p_{\Omega_{1}}v_{\Omega_{1}}\right)do\\
 & = & -\int_{\dot{\Omega}_{0}\cap\dot{\Omega}_{1}}v_{\Omega_{0}}^{\top}T_{\Omega_{0}}n_{\dot{\Omega}_{0}}do-\int_{\dot{\Omega}_{0}\cap\dot{\Omega}_{1}}v_{\Omega_{1}}^{\top}\left(p_{\Omega_{1}}n_{\dot{\Omega}_{0}}\right)do.
\end{eqnarray*}
Since $\left(v_{\Omega_{0}}+v_{\Omega_{1}}\right)\in D\left(\interior\Grad\right)$
is by construction admissible we may assume that $v_{\Omega_{0}}=v_{\Omega_{1}}$
on the interface and conclude that
\begin{equation}
T_{\Omega_{0}}n_{\dot{\Omega}_{0}}+p_{\Omega_{1}}n_{\dot{\Omega}_{0}}=0\label{eq:trans1}
\end{equation}
 is a needed transmission condition. In particular, we see that 
\[
n_{\dot{\Omega}_{0}}\times T_{\Omega_{0}}n_{\dot{\Omega}_{0}}=0.
\]
Inserting the explicit transmission condition (\ref{eq:trans1}) now
yields
\begin{eqnarray*}
0 & = & \int_{\dot{\Omega}_{0}\cap\dot{\Omega}_{1}}\left(v_{\Omega_{0}}-v_{\Omega_{1}}\right)^{\top}\left(p_{\Omega_{1}}n_{\dot{\Omega}_{0}}\right)do\\
 & = & \int_{\dot{\Omega}_{0}\cap\dot{\Omega}_{1}}p_{\Omega_{1}}n_{\dot{\Omega}_{0}}^{\top}\left(v_{\Omega_{0}}-v_{\Omega_{1}}\right)do
\end{eqnarray*}
which, with $p_{\Omega_{1}}$ being arbitrary, now implies
\[
n_{\dot{\Omega}_{0}}^{\top}v_{\Omega_{0}}=n_{\dot{\Omega}_{0}}^{\top}v_{\Omega_{1}}
\]
i.e. the continuity of the normal components 
\[
v_{\Omega_{0},\mathrm{n}}=v_{\Omega_{1},\mathrm{n}},
\]
as a complementing transmission condition. These more or less heuristic
considerations motivate to take the above evo-system as a appropriate
generalization to cases, where the boundary does \emph{not} have a
reasonable normal vector field.

All in all, we summarize our findings in the following well-posedness
result.
\begin{thm}
If $\rho_{*,\Omega_{0}},C_{\Omega_{0}}$ and $\kappa_{\Omega_{1}},c_{\Omega_{1}}$
are selfadjoint, strictly positive definite, continuous operators
on $L^{2}\left(\Omega_{0},\mathbb{R}^{3}\right)$, $L^{2}\left(\Omega_{0},\mathrm{sym}\left[\mathbb{R}^{3\times3}\right]\right)$,
and on $L^{2}\left(\Omega_{1},\mathbb{R}^{3}\right)$, $L^{2}\left(\Omega_{1},\mathbb{R}\right)$,
respectively, the evo-system (\ref{eq:evo-mix}) is Hadamard well-posed.
Moreover, the solution depends causally on the data. \newpage{}
\end{thm}
\begin{rem}
~
\end{rem}
\begin{enumerate}
\item Since $M\left(0\right)\gg0$, we could construct a fundamental solution
of $\partial_{0}+\sqrt{M\left(0\right)}^{-1}A\sqrt{M\left(0\right)}^{-1}$,
which in turn is obtained from the unitary group 
\[
\left(\exp\left(-t\,\sqrt{M\left(0\right)}^{-1}A\sqrt{M\left(0\right)}^{-1}\right)\right)_{t\in\mathbb{R}}
\]
as described above.
\item We note that we may actually allow for completely general \textendash{}
say, for simplicity, rational \textendash{} material laws as long
as condition (\ref{eq:posdef-rat}) is warranted. The above simple
choice has been used as a more approachable illustrating example,
which links up more explicitly with cases considered elsewhere.
\end{enumerate}
\bibliographystyle{siamplain}

\begin{thebibliography}{10}
\bibitem{NME:NME1669} B.~Flemisch, M.~Kaltenbacher, and B.~I.
Wohlmuth. \newblock Elasto-acoustic and acoustic-acoustic coupling
on non-matching grids\textquotedbl{}. \newblock {\em International
Journal for Numerical Methods in Engineering}, 67(13):1791\textendash 1810,
2006.

\bibitem{Ref166} K.~O.~Friedrichs. \newblock Symmetric hyperbolic
linear differential equations. \newblock {\em Comm. Pure Appl.
Math. 7, 345-392}, 1954.

\bibitem{CPA:CPA3160110306} K.~O.~Friedrichs. \newblock Symmetric
positive linear differential equations. \newblock {\em Communications
on Pure and Applied Mathematics}, 11(3):333\textendash 418, 1958.

\bibitem{doi:10.1137/16M1090326} Y. Gao, P. Li, and B.~Zhang. \newblock
Analysis of transient acoustic-elastic interaction in an unbounded
structure. \newblock {\em SIAM Journal on Mathematical Analysis},
49(5):3951\textendash 3972, 2017.

\bibitem{0963.35043} G.~C.~Hsiao, R.~E.~Kleinman, and G.~F.~Roach.
\newblock {Weak solutions of fluid-solid interaction problems.}
\newblock {\em Math. Nachr.}, 218:139\textendash 163, 2000.

\bibitem{KANG2017686} F.~Kang and X.~Jiang. \newblock Variational
approach to shape derivatives for elasto-acoustic coupled scattering
fields and an application with random interfaces. \newblock {\em
Journal of Mathematical Analysis and Applications}, 456(1):686 \textendash{}
704, 2017.

\bibitem{10.2307/93996} H.~Lamb. \newblock On the vibrations of
an elastic plate in contact with water. \newblock {\em Proceedings
of the Royal Society of London. Series A, Containing Papers of a Mathematical
and Physical Character}, 98(690):205\textendash 216, 1920.

\bibitem{doi:10.1121/1.1916256} Melvin Lax. \newblock The effect
of radiation on the vibrations of a circular diaphragm. \newblock
{\em The Journal of the Acoustical Society of America}, 16(1):5\textendash 13,
1944.

\bibitem{Luke:1995:FIA:214875.214884} C.~J. Luke and P.~A. Martin.
\newblock Fluid-solid interaction: Acoustic scattering by a smooth
elastic obstacle. \newblock {\em SIAM J. Appl. Math.}, 55(4):904\textendash 922,
August 1995.

\bibitem{MMA:MMA3866} A.~J. Mulholland, R.~Picard, S.~Trostorff,
and M.~Waurick. \newblock On well-posedness for some thermo-piezoelectric
coupling models. \newblock {\em Mathematical Methods in the Applied
Sciences}, 39(15):4375\textendash 4384, 2016. \newblock mma.3866.

\bibitem{zbMATH05185539} D.~{Natroshvili}, D.~{Sadunishvili},
I.~{Sigua}, and Z.~{Tediashvili}. \newblock {Fluid-solid interaction:
acoustic scattering by an elastic obstacle with Lipschitz boundary.}
\newblock {\em {Mem. Differ. Equ. Math. Phys.}}, 35:91\textendash 127,
2005.

\bibitem{zbMATH03315043} W.~{Nowacki}. \newblock {Some theorems
of asymmetric thermoelasticity.} \newblock {\em {J. Math. Phys.
Sci.}}, 2:111\textendash 122, 1968.

\bibitem{zbMATH03420318} W.~{Nowacki}. \newblock {Dynamische
Probleme der unsymmetrischen Elastizität.} \newblock {Prikl. Mekh.
6, No.4, 31-50)}, 1970.

\bibitem{Nowacki1986} W.~Nowacki. \newblock {\em {Theory of
asymmetric elasticity. Transl. from the Polish by H. Zorski.}} \newblock
{Oxford etc.: Pergamon Press; Warszawa: PWN-Polish Scientific Publishers.
VIII, 383 p.}, 1986.

\bibitem{Pi2009-1} R.~Picard. \newblock {A Structural Observation
for Linear Material Laws in Classical Mathematical Physics.} \newblock
{\em {Math. Methods Appl. Sci.}}, 32(14):1768\textendash 1803,
2009.

\bibitem{Mother2013} R.~Picard. \newblock {Mother Operators and
their Descendants}. \newblock {\em J. Math. Anal. Appl.}, 403(1):54\textendash 62,
2013.

\bibitem{zbMATH06250993} R.~{Picard}. \newblock {Mother operators
and their descendants.} \newblock {\em {J. Math. Anal. Appl.}},
403(1):54\textendash 62, 2013.

\bibitem{Mother2012} R.~Picard. \newblock {Mother Operators and
their Descendants}. \newblock Technical report, TU Dresden, arXiv:1203.6762v2,
2012.

\bibitem{PDE_DeGruyter} R.~Picard and D.~F. McGhee. \newblock
{\em Partial Differential Equations: A unified Hilbert Space Approach},
volume~55 of {\em {De Gruyter Expositions in Mathematics}}.
\newblock {De Gruyter. Berlin, New York. 518 p.}, 2011.

\bibitem{Picard20164888} R.~Picard, St. Seidler, S.~Trostorff,
and M.~Waurick. \newblock On abstract grad-div systems. \newblock
{\em Journal of Differential Equations}, 260(6):4888 \textendash{}
4917, 2016.

\bibitem{ZAMM:ZAMM201300297} R.~Picard, S.~Trostorff, and M.~Waurick.
\newblock On some models for elastic solids with micro-structure.
\newblock {\em ZAMM - Journal of Applied Mathematics and Mechanics
/ Zeitschrift f{ü}r Angewandte Mathematik und Mechanik}, 95(7):664\textendash 689,
2015.

\bibitem{Picard2015} R.~Picard, S.~Trostorff, and M.~Waurick.
\newblock {\em Well-posedness via Monotonicity \textendash{} an
Overview}, volume 250, pages 397\textendash 452. \newblock Springer
International Publishing, Cham, 2015.

\bibitem{SannaMonkola2011} {S. Mönkölä}. \newblock Numerical simulation
of fluid-structure interaction between acoustic and elastic waves.
\newblock {\em {Jyväskylä Studies in Computing}}, 133:136 p.,
2011.

\bibitem{doi:10.1121/1.397156} A.~F. Seybert, T.~W. Wu, and X.~F.
Wu. \newblock Radiation and scattering of acoustic waves from elastic
solids and shells using the boundary element method. \newblock {\em
The Journal of the Acoustical Society of America}, 84(5):1906\textendash 1912,
1988.

\bibitem{WILCOX20109373} L.~C.~Wilcox, G.~Stadler, C.~Burstedde,
and O.~Ghattas. \newblock A high-order discontinuous galerkin method
for wave propagation through coupled elastic-acoustic media. \newblock
{\em Journal of Computational Physics}, 229(24):9373\textendash 9396,
2010. 
\end{thebibliography}

\end{document}